\documentclass[conference]{IEEEtran}
\ifCLASSINFOpdf
\else
\fi
\usepackage{amsmath,dsfont,bbm,epsfig,amssymb,amsfonts,amstext,verbatim,amsopn,cite,subfigure,multirow,multicol,lipsum,xfrac}
\usepackage{amsthm}
\usepackage{mathtools,amsthm}
\usepackage{perpage}
\usepackage{balance}
\usepackage{url}
\usepackage{amsfonts}
\usepackage{epsfig}
\usepackage[font={small}]{caption}
\usepackage{psfrag}	
\usepackage{etoolbox}
\usepackage{algorithmicx}
\usepackage[Algorithm,ruled]{algorithm}
\usepackage{algpseudocode}
\usepackage{pifont}
\usepackage[utf8]{inputenc}
\usepackage[T1]{fontenc}  
\usepackage[nolist]{acronym}
\MakePerPage{footnote}
\usepackage{textcomp}
\usepackage{paralist}
\usepackage{enumitem}
\usepackage{bbm}
\usepackage[process=auto]{pstool}
\usepackage{tikz,pgfplots}
\usetikzlibrary{shapes,arrows}

\newcommand{\setX}{\mathbb{X}}
\newcommand{\setM}{\mathbb{M}}

\newcommand{\setR}{\mathbb{R}}

\newcommand{\setU}{\mathbb{U}}

\newcommand{\her}{\mathsf{H}}

\newcommand{\man}{\mathcal{N}}

\newcommand{\bzero}{{\boldsymbol{0}}}

\newcommand{\bx}{{\boldsymbol{x}}}

\newcommand{\ee}{\mathrm{e}}
\newcommand{\yy}{\mathrm{y}}
\newcommand{\byy}{\mathbf{y}}
\newcommand{\zz}{\mathrm{z}}

\newcommand{\set}[1]{\left\lbrace#1\right\rbrace}
\newcommand{\brc}[1]{\left( #1 \right)}

\newcommand{\bz}{{\boldsymbol{z}}}
\newcommand{\ba}{{\mathbf{a}}}

\newcommand{\Thr}{{\mathrm{Th}}}

\newcommand{\dif}{\mathrm{d}}

\newcommand{\by}{{\boldsymbol{y}}}

\newcommand{\trp}{\mathsf{T}}

\newcommand{\mA}{\mathbf{A}}
\newcommand{\mW}{\mathbf{W}}

\newcommand{\mI}{\mathbf{I}}

\newcommand{\mU}{\mathbf{U}}

\newcommand{\mV}{\mathbf{V}}

\newcommand{\mY}{\mathbf{Y}}

\newcommand{\E}{\mathbb{E}\hspace{.5mm}}

\newcommand{\norm}[1]{\lVert #1 \rVert}

\newcommand{\abs}[1]{\lvert #1 \rvert}

\newtheoremstyle{mystyle}
  {}
  {}
  {}
  {}
  {\bfseries}
  {:}
  { }
  {}
\theoremstyle{mystyle}

\newtheorem{definition}{Definition}
\newtheorem{proposition}{Proposition}

\newtheorem{assumption}{Conjecture}
\algnewcommand\algorithmicLet{\textbf{Let}}
\algnewcommand\Let{\item[\algorithmicLet]}
\algnewcommand\algorithmicSet{\textbf{Set}}
\algnewcommand\Set{\item[\algorithmicSet]}

\algnewcommand\algorithmicInitiate{\textbf{Initiate}}
\algnewcommand\Initiate{\item[\algorithmicInitiate]}
\algnewcommand\algorithmicStart{\textbf{Begin}}
\algnewcommand\Begin{\item[\algorithmicStart]}
\algnewcommand\algorithmicEnd{\textbf{End}}
\algnewcommand\End{\item[\algorithmicEnd]}

\algnewcommand\algorithmicOutP{\textbf{Output:}}
\algnewcommand\Out{\item[\algorithmicOutP]}

\algnewcommand\algorithmicInP{\textbf{Input:}}
\algnewcommand\In{\item[\algorithmicInP]}

\newcounter{bar}

\begin{document}
\title{Oversampled Adaptive Sensing with Random Projections: Analysis and Algorithmic Approaches\vspace*{-2mm}}

\author{
\IEEEauthorblockN{
Ralf R. M\"uller\IEEEauthorrefmark{1},
Ali Bereyhi\IEEEauthorrefmark{1},
Christoph F. Mecklenbr\"auker\IEEEauthorrefmark{2}\thanks{This work has been accepted for presentation in the 18th IEEE International Symposium on Signal Processing and Information Technology (ISSPIT) 2018 in Louisville, Kentucky, USA. The link to the final version in the Proceedings of ISSPIT will be available later.}
}
\IEEEauthorblockA{
\IEEEauthorrefmark{1}Institute for Digital Communications, Friedrich-Alexander Universit\"at Erlangen-N\"urnberg, Germany\\
\IEEEauthorrefmark{2}Institute of Telecommunications, Technische Universit\"at Wien, Austria\\
ralf.r.mueller@fau.de, ali.bereyhi@fau.de, cfm@nt.tuwien.ac.at\vspace*{-2mm}
}
}


\IEEEoverridecommandlockouts

\maketitle

\begin{acronym}
\acro{oas}[OAS]{oversampled adaptive sensing}
\acro{csi}[CSI]{Channel State Information}
\acro{awgn}[AWGN]{Additive White Gaussian Noise}
\acro{iid}[i.i.d.]{independent and identically distributed}
\acro{ut}[UT]{User Terminal}
\acro{bs}[BS]{Base Station}
\acro{tas}[TAS]{Transmit Antenna Selection}
\acro{lse}[LSE]{Least Squared Error}
\acro{rhs}[r.h.s.]{right hand side}
\acro{lhs}[l.h.s.]{left hand side}
\acro{wrt}[w.r.t.]{with respect to}
\acro{rs}[RS]{Replica Symmetry}
\acro{rsb}[RSB]{Replica Symmetry Breaking}
\acro{mse}[MSE]{mean squared error}
\acro{mmse}[MMSE]{minimum MSE}
\acro{rzf}[RZF]{Regularized Zero Forcing}
\acro{snr}[SNR]{Signal-to-Noise Ratio}
\acro{sinr}[SINR]{signal to interference and noise ratio}
\acro{rf}[RF]{Radio Frequency}
\acro{mf}[MF]{Match Filtering}
\end{acronym}

\begin{abstract}
Oversampled adaptive sensing (OAS) is a recently proposed Bayesian framework which sequentially adapts the sensing basis. In OAS, estimation quality is, in each step,~measured by conditional mean squared errors (MSEs), and~the~basis for the next sensing step is adapted accordingly.~For~given\\ average sensing time, OAS reduces the MSE compared~to~non- adaptive schemes, when the signal is sparse. This~paper~studies the asymptotic performance of Bayesian OAS,~for~unitarily~in- variant random projections. For sparse signals, it is shown that OAS with Bayesian recovery~and~hard~adaptation significantly outperforms the minimum MSE bound for non-adaptive sensing. To address implementational aspects, two computationally tractable algorithms are proposed, and their performances are compared against the state-of-the-art non-adaptive algorithms via numerical simulations. Investigations depict that these low-complexity OAS algorithms, despite their suboptimality,~outperform well-known non-adaptive schemes for sparse recovery, such as LASSO, with rather small oversampling factors. This gain grows, as the compression rate increases.\vspace*{-2mm}
\end{abstract}

\IEEEpeerreviewmaketitle

\section{Introduction}
\label{sec:intro}
In \ac{oas} \cite{muller2018oversampled}, $K$ sensors observe $N$ signal samples linearly. The  projections from the samples to the observations are modifiable and corrupted with noise. The array is supposed to sense the signal within a fixed time duration of length $T_{\rm s}$. The main objective is to design a sensing procedure, such that the signal is recovered from observations with high fidelity. This is a classical problem and has been widely studied in various contexts \cite{seber2012linear,donoho2006compressed,candes2006near,foucart2013mathematical}. The conventional approaches proposed in signal processing and information theory are formulated in the following generic form:
\begin{enumerate}[label=(\alph*)]
\item The required number of observations $L$ is determined based on prior information on the signal. For example, $L=N$ when the signal is assumed to be uniformly~distributed. For sparse signals, however, $L < N$. 
\item The oversampling factor is calculated as $M = \lceil L/K \rceil$. The array needs to sense the samples $M$ times, in order to collect as much observations as required.
\item Sensing duration $T_{\rm s}$ is divided into $M$ subframes~each of length $T_{\rm m} = T_{\rm s} / M$. The array observes the samples in each of these subframes using different projections.
\item The samples are estimated via a recovery scheme from the $MK$ collected measurements.
\end{enumerate}

For this classical framework, theoretical discussions are roughly divided into two main directions:
\begin{inparaenum}
\item A body of work investigates the number of required observations which guarantees a certain level of estimation quality;~see~for~example discussions in \cite{donoho2006compressed,candes2006near,foucart2013mathematical,wu2010renyi}. 
\item Another group of analytic studies characterize theoretical bounds on the performance of various recovery schemes, e.g. \cite{bereyhi2016statistical,bereyhi2017replica,rangan2012asymptotic,vehkapera2016analysis}.
\end{inparaenum}
Algorithmic~approaches, on another hand, mainly focus on the design of computationally tractable recovery algorithms and on the construction of linear mixing for efficient observation.

\subsection{Sequential Sensing via \ac{oas}}
\ac{oas} deviates from the classic framework by developing an \textit{adaptive} technique for sequential sensing. In this~scheme, the signal is \textit{oversampled} by an arbitrary factor, and linear projections are updated in each step via conditional posterior distortions calculated in the previous subframe. This sequentially adaptive approach was widely believed to be ineffectual for $M > \lceil L/K \rceil$ following the fact that the growth in oversampling factor reduces the duration of each subframe, i.e. $T_{\rm m}$, and hence increases the noise power in each individual sensing. Investigations have demonstrated that while this belief is true for signals with absolutely continuous priors, sequential adaptation is in fact beneficial when the signal has a mixture prior; see discussions in \cite{muller2018oversampled} and \cite{haupt2009adaptive}. This is illustrative by considering an example from sparse recovery. When the signal is sparse, zero samples are recovered with high reliability in initial subframes and canceled out in subsequent subframes by adaptation. This reduces the dimensionality of the problem in the subsequent subframes and improves the performance.

\subsection{Contributions}
The initial study on \ac{oas} in \cite{muller2018oversampled} considered scenarios with orthogonal and deterministic projections. Such an assumption was primarily taken for sake of analytical tractability. Nevertheless, practical scenarios often deal with conditions in~which observations are acquired through non-orthogonal random projections. We address this issue by investigating the performance of \ac{oas} for unitarily invariant random~pro- jections. In this respect, the main contributions are
\begin{enumerate}
\item \textit{Asymptotic characterization of \ac{oas}:} Using the decoupling property of Bayesian estimators, we characterize the performance of \ac{oas} for unitarily invariant matrices in the large-system limit. The analytic result enables us to quantify the gain over non-adaptive schemes.
\item \textit{Algorithmic approaches:} We propose computationally tractable algorithms for \ac{oas} in which conditional distortions are derived with low complexity. Even~though these low-complexity implementations degrade the performance, the algorithms are shown~to outperform~the state-of-the-art non-adaptive schemes, even for rather small oversampling factors.
\end{enumerate}

\subsection{Notation}
We represent scalars, vectors and matrices with non-bold, bold lower case and bold upper case letters, respectively. A $K \times K$ identity matrix is shown by $\mI_K$. $\mA^{\trp}$ indicates the transpose of $\mA$. The set of real numbers is denoted by $\setR$. The expectation operator is identified by $\E $. We use the shortened notation $[N]$ to represent $\set{1, \ldots , N}$.

\section{Bayesian OAS Framework}
\label{sec:sys}
We consider a generic form of \ac{oas} as follows:
\begin{enumerate}[label=(\alph*)]
\item The vector $\bx\in\setX^N$ is sensed $M$ times.
\item In step $m$, the matrix of stacked observations $\mY_m$ is
\begin{align}
\mY_m \coloneqq \left[ \by_1, \ldots, \by_m\right]
\end{align}
where $\by_m\in\setR^{K}$ for $m\in[M]$.
\item The vector of measurements in step $m\in[M]$ reads
\begin{align}
\by_m=\mA_m \mW_m \bx + \bz_m
\end{align}
where $\mA_m\in \setR^{K\times N}$ denotes the linear projection, and $\mW_m  \in \setR^{N\times N}$ is a diagonal matrix containing adaptation coefficients in subframe $m$. $\bz_m$ moreover denotes measurement noise and is assumed to be \ac{iid} Gaussian with mean zero. Denoting the noise variance for sensing duration $T_{\rm s}$ with $\sigma^2$, the variance in each subframe is $M\sigma^2$; hence,  $\bz_m\sim \man\brc{\bzero, M \sigma^2 \mI_K}$.
\item The $n$-th sample, i.e. $x_n$, is recovered in step $m$ from $\mY_m$ via a Bayesian estimator. Denoting the recovered sample by $r_n\brc{ \mY_m }$, this means that
\begin{align}
r_n\brc{ \mY_m } = \E\set{x_n|\mY_m}
\end{align}
for some postulated prior distribution $q\brc{x}$. This estimator reduces to several recovery algorithms, e.g.~the linear estimator, LASSO and optimal Bayesian estimator, for specific choices of $q\brc{x}$.
\item In step $m$, the conditional average distortion of sample $n$, with respect to distortion function $D[\cdot;\cdot]$, is 
\begin{align}
d_n\left(\mY_m\right) \coloneqq \E_{x_n \vert \mY_m} D\left[x_n; r_n(\mY_m) \right].
\end{align}
\item $\mW_m$ is determined as a function of the conditional~average distortions in step $m-1$, i.e.
\begin{align}
\mW_m = f_\Thr \brc{ d_1\left(\mY_{m-1}\right), \ldots,d_N\left(\mY_{m-1}\right)  }.
\end{align}
Examples for $f_\Thr(\cdot)$ are hard and soft adaptations.
\end{enumerate}
\subsection{Performance Analysis}
We consider the case in which $K$ and $N$ grow large such that the inverse load, i.e. $\rho = {N}/{K}$, remains constant. It is moreover assumed that
\begin{inparaenum}
\item $\mW_m$ is constructed by hard~adap- tation, meaning that diagonal entries are either one or zero. 
\item $\mA_m$ is bi-unitarily invariant, meaning that it has the same joint distribution as $\mU \mA_m \mV$, for any unitary matrices $\mU$ and $\mV$, such that $\mU$, $\mV$, and $\mA_m$ are jointly independent.
\end{inparaenum}

To quantify the performance, we consider the average \textit{per-sample} distortion as the measure which is defined as
\begin{align*}
D_{\rm avg} = \frac{1}{N}\E \sum_{n=1}^N D\left[x_n; r_n(\mY_M) \right] \label{Distortion_per}
\end{align*}
for the given distortion function $D[\cdot;\cdot]$.
\subsection{Some Definitions}
For sake of brevity, we define the following parameters.
\begin{itemize}
\item $\setM_m \subseteq [N]$ contains the indices of all non-zero~diago- nal entries in $\mW_m$. In other words, $\setM_m$ is an index set of the samples which are sensed in subframe $m$.
\item The indexer function $i_n \brc{\cdot} $ is defined as
\begin{align}
i_n \brc{m} = \begin{cases}
m & \text{if} \  n\in\setM_m\\
\epsilon & \text{if} \  n\notin\setM_m
\end{cases},
\end{align}
for some error symbol $\epsilon$.
\item $\setU_n \brc{m} \subseteq [m]$ contains the subframes in which $x_n$ is sensed, i.e., %
$\setU_n \brc{m} = \set{ u \in [m] : \ i_n \brc{u} \neq \epsilon }$.
\end{itemize}

\section{Asymptotic Characterization of OAS}
\label{sec:application}
The direct approach for the analysis of \ac{oas} requires~conditional average distortions to be derived explicitly for each subframe. This is not tractable for various choices of the prior distribution. We hence invoke the decoupling property of Bayesian estimators following discussions in \cite{bereyhi2016rsb,rangan2012asymptotic,bereyhi2016statistical} and the references therein.
\subsection{Decoupling Principle}
\hspace*{-2mm}The decoupling principle states that $\brc{ r_n\brc{\by_m}, x_n} $ for~$n\in$ $\setM_m$ converges in distribution to %
$\brc{r_n \brc{\yy_{n}{[m]}},x_n}$ 
where
\begin{align}
\yy_{n} \left[ m \right] = x_n + \zz_{n} \left[ m \right]  . \label{eq:decop_m}
\end{align}
with $\zz_{n}{[m]}\sim\man\brc{0, M \sigma_{m}^2}$. Here, $\sigma_{m}^2$ is the effective~noise variance for sensing duration $T_{\rm s}$ with sensing matrix $\mA_m$. An explicit derivation of $\sigma_{m}^2$ in terms of $\sigma^2$ and the statistics of $\mA_m$ is given in \cite{bereyhi2016statistical,rangan2012asymptotic}.

Symbol $\yy_{n} \left[ m \right]$ is often called the \textit{decoupled output} and is a visualization of asymptotic Gaussianity of interference. %
The \textit{decoupled} scalar Bayesian estimator, i.e.
\begin{align}
r_n \brc{\yy_{n}{[m]}} = \E \set{x_n \vert \yy_n[m] },
\end{align}
recovers sample $x_n$ by postulating that $\yy_n[m]$ is related to $x_n$ via \eqref{eq:decop_m}, and that sample $x_n$ is distributed with $q\brc{x}$.~The decoupling principle implies that as $N$ and $K$ grow~large, the distribution of the true recovered sample, i.e., $r_n\brc{\by_m}$, conditioned to $x_n$ converges to that of the sample recovered from the $n$-th decoupled output, i.e., $r_n\brc{\yy_n [m]}$.
%


\subsection{Asymptotics via the Decoupling Principle}
%
Using the decoupling property of Bayesian estimators, we derive a heuristic bound on the asymptotic performance of \ac{oas}. To this end, consider an \ac{oas} setting with~$M$~subframes. At subframe $m$, we define for $n\in[N]$
\begin{align}
\byy_n[m]= 
\left[ \yy_n\left[i_n(1)\right], \ldots, \yy_n\left[i_n(m)\right] \right]^\trp
\end{align}
where $\yy_n[i_n(\hat{m})]$ for $n\in\setM_{\hat{m}}$ is generated according to \eqref{eq:decop_m} and $\yy_n[\epsilon]\coloneqq 0$. $\byy_n[m]$ contains the decoupled outputs of the $n$-th sample, up to sensing step $m$, from those subframes in which $x_n$ is sensed. We now define the \textit{stacked decoupled system} in subframe $m$ as follows: 
\begin{definition}
For $m \in [M]$, the stacked decoupled system in subframe $m$ consists of signal samples $x_n$, observations 
\begin{align}
\bar{\yy}_n[m] = \sum_{u=1}^m \yy_n[i_n(u)] = \sum_{u\in\setU_n \brc{m}} \yy_n[u],
\end{align}
and recovered symbols $r_n\brc{\bar{\yy}_n[m]}$, for $n\in[N]$.
\end{definition}

Proposition~\ref{prop:1} indicates that \ac{oas} is asymptotically characterized via the stacked decoupled system. To state this~result, let us define the degraded version of an \ac{oas} setting.

\begin{definition}[Degraded \ac{oas} setting]
\label{assump:1}
Consider a Bayesian \ac{oas} setting in which sample $n$ in each subframe is reconstructed by recovery algorithm $r_n\brc{\cdot}$, and the conditional distortion is determined via $D[\cdot;\cdot]$. The degraded version of this \ac{oas} setting adapts $\mW_{m+1}$ using
\begin{align}
\hat{d}_n\left(\mY_m\right) \coloneqq \E_{x_n \vert r_n\brc{\mY_m}} D\left[x_n; r_n(\mY_m) \right].
\end{align}
%
\end{definition}
The degraded setting assumes that $r_n\brc{\mY_m}$ is a sufficient statistics for estimating $D\left[x_n; r_n(\mY_m) \right]$. This assumption in general can degrade the estimation performance. Hence, the average distortion of this setting, in general, bounds the average distortion of the original \ac{oas} setup from above.

Noting that $\hat{d}_n\left(\mY_m\right)$ is only a function of $r_n\brc{\mY_m}$, one can use the decoupling principle as show that the marginal distribution of $\hat{d}_n\left(\mY_1\right)$ does not depend on $n$; see \cite{bereyhi2016rsb}. We further consider the following conjecture:
\begin{assumption}
\label{conj1}
As $N\uparrow\infty$, $\{\hat{d}_n\left(\mY_1\right):n\in [N]\}$ is ergodic.
\end{assumption}
Conjecture~\ref{conj1} assumes that the empirical average over a function of $\hat{d}_n\left(\mY_m\right)$ converges to the expectation over the marginal distribution\footnote{The validity of the conjecture is straightforwardly verified~for~some~particular Bayesian estimators.}. Assuming this conjecture to hold, we state Proposition~\ref{prop:1} as follows:
\begin{proposition}
\label{prop:1}
Assume Conjecture~\ref{conj1} holds for recovery~algorithm $r_n\brc{\cdot}$ and distortion function $D[\cdot;\cdot]$. Let $\bx$ be \ac{iid} with $x_n\sim q_X\brc{x}$ and \ac{oas} employ hard thresholding with threshold $D_\Thr$ for adaptation, i.e., at each subframe, those samples are sensed whose conditional distortions are more than $D_{\rm Th}$. Then, in subframe $m$, the performance of the degraded \ac{oas} is equivalent to \ac{oas} performing over the stacked decoupled system with distortion function $D[\cdot;\cdot]$.
\end{proposition}

\begin{proof}
Due to the page limit, we only give a brief sketch of the approach here and leave the details for the extended version of the manuscript. The proof follows from induction: Starting from $m=1$, we have $\setM_1=[N]$. Hence, the $n$-th conditional distortion of the degraded setting reads
\begin{subequations}
\begin{align}
\hat{d}_n\left(\mY_1\right) &= \int D\left[x_n; r_n(\by_1) \right] \dif P\brc{x_n \vert r_n\brc{\by_1}}\\
&\coloneqq F\brc{r_n\brc{\by_1}} \label{def:F}
\end{align}
\end{subequations}
where $F\brc{\cdot}$ is a deterministic function. In the stacked~decoupled setting, $\bar\yy_n[1] = \byy_n [1]$ for $n\in [N]$. Thus, 
\begin{align}
\hspace*{-2mm}\hat{d}_n\left(\bar\yy_n[1]\right) &= \int D\left[x_n; r_n(\yy_n[1]) \right] \dif P\brc{x_n \vert r_n\brc{\yy_n[1]}}.
\end{align}
The decoupling principle indicates convergence of the two settings in distribution. This implies that 
\begin{align}
\dif P \brc{x_n \vert r_n\brc{\by_1}} 
=\dif P \brc{x_n \vert r_n\brc{\yy_n[1]}}.
\end{align}
Consequently, we conclude that $\hat{d}_n\brc{ \bar\yy_n[1] } = F\brc{r_n\brc{\yy_n[1]}}$ with $F\brc{\cdot}$ being defined in \eqref{def:F}. 

Noting that $\bx$ is \ac{iid}, we can infer that $\set{r_n\brc{\yy_n[1]}}$~for $n\in[N]$ is \ac{iid}, too. This implies that $\{\hat{d}_n\left(\mY_1\right)\}$ are~identically distributed. Moreover, Conjecture~\ref{conj1} indicates that the empirical distribution of $\{\hat{d}_n\left(\mY_1\right)\}$ converges to the~distribution of $\hat{d}\brc{r_n\brc{\yy_n[1]}}$. Considering this conclusion, the strong law of large numbers implies that the frequency of entries whose conditional distortions are more than $D_\Thr$ is the same in both settings. Hence, in the asymptotic~regime, both settings choose the same number of samples~for~sensing in the next subframe.

Now consider subframe $m$, and assume $\setM_m$ in both~settings contains the same number of indices. In this case, 
\begin{align}
\hat{d}_n\left(\mY_m\right) 
= F\brc{r_n\brc{\bar\by_m}}
\end{align}
with $\bar\by_m = [\by_1^\trp,\ldots,\by_m^\trp]^\trp$. $\hat{d}_n\left(\mY_m\right)$ is the conditional~distortion of a cascaded setting whose sensing matrix is
\begin{align}
\bar\mA_m =
\begin{bmatrix}
\brc{\mW_1 \mA_1}^\trp &\ldots &\brc{\mW_m \mA_m}^\trp
\end{bmatrix}^\trp.
\end{align}
This cascaded setting can be grouped into $m$ sub-settings, each describing one of the previous subframes. The dimensions of these sub-settings grow large proportional to $N$. Thus, the decoupled setting in this case is described via the set of decoupled outputs of each sub-setting; see discussions in \cite{bereyhi2018ICASSP}. As a result, $\brc{r_n\brc{\bar\by_m},x_n}$ converges in distribution to $\brc{r_n\brc{\byy_n[m]},x_n}$. By a similar approach as for $m=1$, we could conclude that for $n\in[N]$
\begin{align}
\hat{d}_n\left(\byy_n[m]\right) = \E_{x_n \vert r_n\brc{\byy_n[m]}}\set{D\left[x_n; r_n\brc{\byy_n[m]} \right]}
\end{align}
is \ac{iid} whose distribution is identical to that of $\hat{d}_n\left(\mY_m\right)$. The Fisher–Neyman factorization theorem \cite{hogg1995introduction} states that $\bar{\yy}_n[m]$ is a sufficient statistics of $\byy_n[m]$. Thus,
$\hat{d}_n \brc{\byy_n[m]}=$ $\hat{d}_n \brc{\bar{\yy}_n[m]}$. 
Taking the assumption of ergodicity, it is finally concluded that $\setM_{m+1}$ in the both settings contains asymptotically the same number of indices. 
\end{proof}
\subsection{Gain over Non-adaptive Sensing}
Proposition~\ref{prop:1} is utilized to characterize the performance of \ac{oas} with unitarily invariant projections and a Bayesian estimator. In this respect, in each subframe, the stacked decoupled system is realized independently. Although these realizations do not result in exact conditional distortions in each subframe, the average per-symbol distortion is asymptotically equal to the one determined in the degraded \ac{oas} setting. Noting that this distortion is an upper-bound on the average distortion of the \ac{oas} setting, it is concluded that the gain reported via the stacked decoupled setting is in general a lower-bound on the actual gain acquired by using a Bayesian \ac{oas} setting.

In Fig.~\ref{Fig:1}, the \ac{mse} is plotted versus inverse load $\rho$ considering \ac{iid} sensing matrices and the optimal Bayesian estimator. The \ac{mse} is determined by~setting $D\left[x_n; r_n \right] = \brc{x_n-r_n}^2$ in \eqref{Distortion_per}. The signal samples are \ac{iid} Bernoulli-Gaussian, meaning that $x_n=b_n t_n$ with $b_n$ being $\delta$-Bernoulli, i.e. %
$\Pr\set{b_n=1} = 1- \Pr\set{b_n=0} = \delta$, %
and $t_n \sim \man\brc{0, \sigma_t^2}$. In the figure, $\delta = 0.1$ and $\sigma_t^2=1$. The noise variance is set to $\sigma^2=0.01$. The adaptation is done by hard thresholding with $\log D_{\rm Th} = -26.5$ $\rm dB$. For sake of comparison, the \ac{mmse} for \textit{non-adaptive}\footnote{This means $M=1$} sensing is plotted, too.  As the figure depicts, at larger inverse loads, the \ac{oas} significantly outperforms the \ac{mmse} bound for non-adaptive sensing.
\begin{figure}[t]
\centering
\input{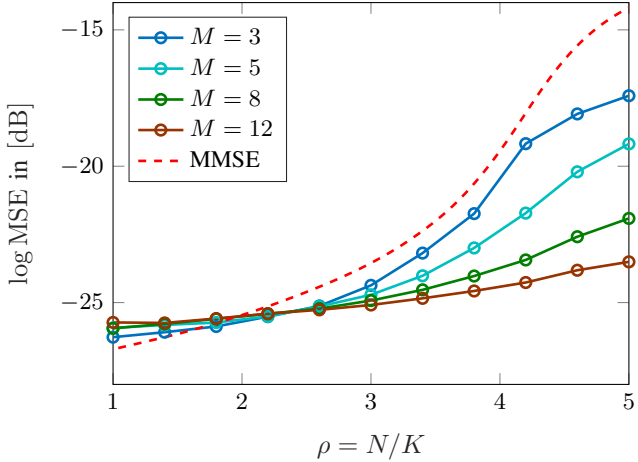}
\caption{\ac{mse} vs the inverse load. Signal samples are sensed by \ac{iid} random matrices and recovered by optimal Bayesian estimator. The source is considered to be sparse Gaussian with sparsity $\delta=0.1$. The \ac{oas} is adapted by hard thresholding with $\log D_{\rm Th} = -26.5$ $\rm dB$. The noise variance is set to $\sigma^2 = 0.01$. The dashed red line denotes the \ac{mmse} for non-adaptive optimal Bayesian recovery.\vspace*{-3mm}}
\label{Fig:1}
\end{figure}

\section{Low-complexity Algorithms for OAS}
\label{sec:algorithm}
The computational complexity of \ac{oas} is mainly dominated by the task of calculating conditional distortions. This task is intractable for the optimal Bayesian approach. In the sequel, we propose suboptimal low-complexity algorithms for \ac{oas} and compare their performances against the state~of the art. It is shown that, despite the suboptimality, the~algorithms still outperform the non-adaptive benchmark.

\subsection{\ac{oas} via Orthogonal Observations}
In this approach, the sensing matrix in each subframe is generated randomly according to the Haar distribution\footnote{A Haar matrix is uniformly distributed over the orthogonal group.}. Due to orthogonality of the sensing weights, the observations are simply decoupled, and hence the conditional distortions are calculated tractably via Bayes' theorem. The approach~is~illustrated in Algorithm~\ref{alg1}. For adaptation, the worst-case hard adaptation is employed. This means that in each subframe, the $K$ signal samples with largest conditional distortions in the previous subframe are sensed.
\subsubsection*{Derivation of the Algorithm}
Let $\setM_m = \set{j_1, \ldots, j_K}$ be the set of $K$ samples sensed in subframe $m$. We generate a Haar matrix $\mU_m \in \setR^{K\times K}$ and construct $\mA_m$ by setting~its $j_k$-th column to the $k$-th column of $\mU_m$. $\mA_m \mW_m$ is hence a $K\times M$ matrix whose columns are equal to the columns of $\mU_m$, if indexed by $\setM_m$, and zero otherwise. 

To decouple the observations, we multiply them with the transposed projection. In this case, for $k\in[K]$, we have
\begin{align}
\yy_{j_k} [m] &=  [\mU_m^\trp \by_m]_k = x_{j_k}  + \tilde{z}_{j_k} \brc{m} 
\end{align}
where $\tilde{\bz}\brc{m} = \mU_m^\trp \bz_m$. Since $\mU_m^\trp$ is orthogonal, we can conclude that $\tilde{z}_{j_k} \brc{m}\sim \man\brc{0,M\sigma^2}$.

The stacked decoupled output in this case reads
\begin{align}
\bar{\yy}_{j_k} [m] 
&= \abs{\setU_{j_k} \brc{m}} \ x_{j_k}  + \tilde{\zz}_{j_k} \brc{m}
\end{align}
where $\tilde{\zz}_{j_k} \brc{m}\sim \man\brc{0,\abs{\setU_{j_k} \brc{m}} M\sigma^2}$. Consequently, the Bayesian estimator for $x_n$ is given by
\begin{align}
r_m\brc{\mY_m} = \int x_n \ p\brc{\bar{\yy}_n[m]|x_n} q\brc{x_n} \dif x_n, \label{eq:Bayes1}
\end{align}
and the conditional distortion is determined by \cite{muller2018oversampled}
\begin{align}
d_{n} \brc{\mY_m} = \dfrac{ \partial}{\partial \bar{\yy}_{n} [m]} r_{n} \brc{\mY_m}. \label{eq:Bayes2}
\end{align}
\begin{algorithm}[t]
\caption{\ac{oas} via Orthogonal Observations}
\label{alg1}
\begin{algorithmic}[0]
\Initiate Set 
$d_n \brc{\mY_0} = +\infty$,
$\bar{\yy}_n [0] =0$,
and 
$c_n [0] =0$.
\vspace{1mm}
\For{$m=1:M$}\\
Generate $\mU_m\in\setR^{K\times K}$ from the Haar distribution.\\
Set $\mA_m\in\setR^{K\times N}$ arbitrarily and $\mW_m=\set{0}^{N\times N}$.\\
Find $\set{j_1, \ldots, j_N}$, such that 
\begin{align}
d_{j_1} \brc{\mY_{m-1}} \geq \ldots \geq d_{j_N} \brc{\mY_{m-1}}.
\end{align}
\For{$k\in[K]$}\\
$\qquad$Set
$\mA_m\brc{:,j_k} = \mU_m\brc{:,k}$ and 
$\mW_m\brc{j_k,j_k} = 1$\\
$\qquad$Set $c_{j_k}[m] = c_{j_k}[m-1]+1$
\EndFor\vspace*{1mm}\\
Sense $\bx$ via $\mA_m\mW_m$, i.e. $\by_m = \mA_m\mW_m \bx + \bz_m$.
\For{$k\in[K]$} update $\bar{\yy}_{j_k} [m]$
\begin{subequations}
\begin{align}
\bar{\yy}_{j_k} [m] &= \bar{\yy}_{j_k} [m-1] + [\mU_m^\her \by_m]_k,\\
\ee_{j_k}\brc{x \vert \bar{\yy}_{j_k} [m] } &= \exp\set{-\dfrac{ \brc{ \bar{\yy}_{j_k} [m] \hspace*{-.5mm}-\hspace*{-.5mm} c_{j_k}[m] x }^2 }{2c_{j_k}[m] M\sigma^2}}\hspace*{-.5mm}.
\end{align}
\end{subequations}
$\qquad$Update recoveries and conditional distortions as
\begin{subequations}
\begin{align}
r_{j_k} \brc{\mY_m} = \dfrac{\int x \  \ee_{j_k}\brc{x \vert \bar{\yy}_{j_k} [m] } q\brc{x} \dif x }{\int \ee_{j_k}\brc{x \vert \bar{\yy}_{j_k} [m] } q\brc{x} \dif x}.\\
d_{j_k} \brc{\mY_m} = \dfrac{ \partial}{\partial \bar{\yy}_{j_k} [m]} r_{j_k} \brc{\mY_m}.
\end{align}
\end{subequations}
\EndFor
\EndFor
\end{algorithmic}
\end{algorithm}
\subsubsection*{Numerical Investigations}
Fig.~\ref{Fig:2} shows the performance of Algorithm~\ref{alg1} for recovery of $N=200$ sparse Gaussian samples with $\delta=0.1$ and $\sigma_t^2 =1$. The noise variance is set to $\sigma^2=0.01$. The performance is compared against the asymptotics of non-adaptive LASSO when a row-orthogonal random matrix is employed for sensing; see \cite{bereyhi2016statistical,vehkapera2016analysis}. As the figure depicts, for rather small choices of $M$, the algorithm outperforms non-adaptive LASSO within a large range of $\rho$. This gain increases significantly as $M$ grows.
\subsection{\ac{oas} via Matched Filtering}
Algorithm~\ref{alg1} requires an independent orthogonal basis in each subframe. This can be further avoided by matched~filtering. In this approach, samples are sensed by \ac{iid}~weights, and the calculation of conditional distortions is simplified by postulating the impairment of other samples to be Gaussian. This \ac{oas} approach is given in Algorithm~\ref{alg2}. The algorithm employs hard thresholding for adaptation meaning that, in each subframe, samples whose conditional distortions in~the previous subframe are more than a threshold are sensed.
\begin{figure}[t]
\centering
\input{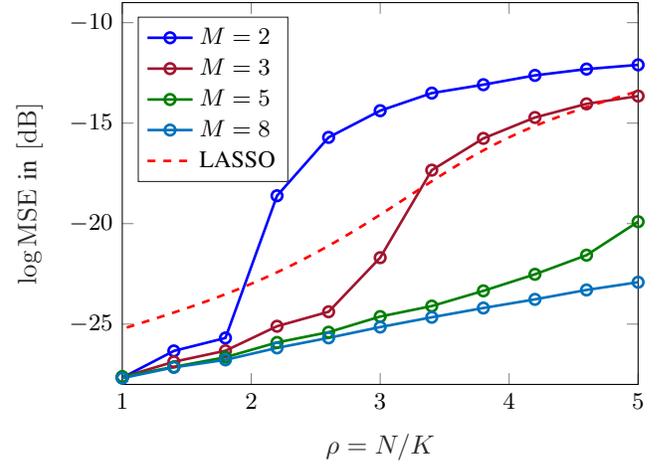}
\caption{\ac{mse} vs $\rho=N/K$ for Algorithm~\ref{alg1}. Signal samples are sparse Gaussian with sparsity $\delta=0.1$, and sensed by orthogonal random matrices. \ac{oas} is adapted by worst-case hard adaptation. The noise variance is set to $\sigma^2 = 0.01$. The dashed red line denotes asymptotic performance for non-adaptive LASSO recovery.\vspace*{-3mm}}
\label{Fig:2}
\end{figure}
\subsubsection*{Derivation of the Algorithm}
Let $\mA_m$ have \ac{iid} entries with zero mean and variance $1/\sqrt{K}$. In subframe $m$, 
\begin{align}
\by_m \hspace*{-.5mm}= \hspace*{-.5mm} \mA_m\mW_m \bx \hspace*{-.5mm} + \hspace*{-.5mm} \bz_m 
\hspace*{-.5mm}=\hspace*{-.5mm} x_n w_n\brc{m} \ba_n \brc{m} \hspace*{-.5mm}+\hspace*{-.5mm} \tilde{\bz}_n \brc{m}
\end{align}
where $\ba_j \brc{m}$ is the $j$-th column of $\mA_m$, $w_j\brc{m}$ represents the $j$-th diagonal entry of $\mW_m$ and\vspace*{-2mm}
\begin{align}
\tilde{\bz}_n \brc{m} = \sum_{j=1, j\neq n}^N  x_j w_j \brc{m} \ba_j \brc{m} + \bz_m.
\end{align}
$\qquad$\vspace*{-3mm}\\
Matched filtering assumes that $\tilde{\bz}_n \brc{m}$ is an \ac{iid} Gaussian vector independent of $x_n w_n\brc{m} \ba_n \brc{m}$. Although such an assumption is not valid in general, it lets us calculate~conditional distortions tractably. By the standard approach, 
\begin{align}
\E\{\tilde{\bz}_n \brc{m} \tilde{\bz}_n \brc{m}^\trp \} = \sigma_{\rm MF}^2 \mI_K,
\end{align}
where 
$\sigma_{\rm MF}^2 \coloneqq \rho_m \sigma_x^2 + M \sigma^2$.
Here, $\rho_m = \brc{\abs{\setM_m}-1}/K$ denotes the inverse load in subframe $m$, and $\sigma_x^2 = \E x_n^2$. 

Assuming Gaussian impairment, recovery of each sample is a scalar estimation problem. The Fisher–Neyman factorization indicates that $\ba_n^\trp\brc{m} \by_m / \norm{\ba_n \brc{m}}^2$ is a Bayesian sufficient statistics of $\by_m$ for estimating $x_n$. Hence, when $x_n$ is sensed in subframe $m$, i.e. for $n\in\setM_m$,~we~approximate the decoupled symbol with 
\begin{align}
\yy_n [m] &= \frac{1}{\norm{\ba_n \brc{m}}^2} \ba_n^\trp\brc{m} \by_m 
= x_n + \zz_n \brc{m}
\end{align}
where
$\zz_n \brc{m} \coloneqq \ba_n^\trp\brc{m} \tilde{\bz}_n \brc{m}/ {\norm{\ba_n \brc{m}}^2}$ %
is distributed as $\man\brc{0,\sigma_{\rm MF}^2 / \norm{\ba_n \brc{m}}^2}$. Consequently,
\begin{align}
\bar{\yy}_n[m] 
&= \abs{\setU_n \brc{m}} \ x_n  + \tilde{\zz}_n \brc{m}
\end{align}
where $\tilde{\zz}_n\brc{m} \sim \man\brc{0,\sigma_n^2 \brc{m}}$ with
\begin{align}
\sigma_n^2 \brc{m} &= \sigma_{\rm MF}^2 \sum_{u\in\setU_n \brc{m}} \frac{1}{\norm{\ba_n \brc{u}}^2}.
\end{align}
Finally, the Bayesian estimator and conditional distortions are derived as in \eqref{eq:Bayes1} and \eqref{eq:Bayes2}.
\subsubsection*{Numerical Investigations}
The \ac{mse} for Algorithm~\ref{alg2} is sketched versus the inverse load in Fig.~\ref{Fig:3}. The~figure shows a degraded performance compared to Algorithm~\ref{alg1}. This is due to the fact that observations in Algorithm~\ref{alg1} are decoupled deterministically in each subframe using orthogonality of the projecting vectors. For sake of comparison, the asymptotic performance of LASSO, as well as the \ac{mmse} bound, for non-adaptive sensing with an \ac{iid} matrix is plotted, too. As it depicts, by $M=12$ subframes, \ac{oas} with matched filtering outperforms LASSO for $\rho \geq 1$. Despite its suboptimality, this approach outperforms the non-adaptive \ac{mmse} bound for large inverse loads.
\begin{figure}[t]
\centering
\input{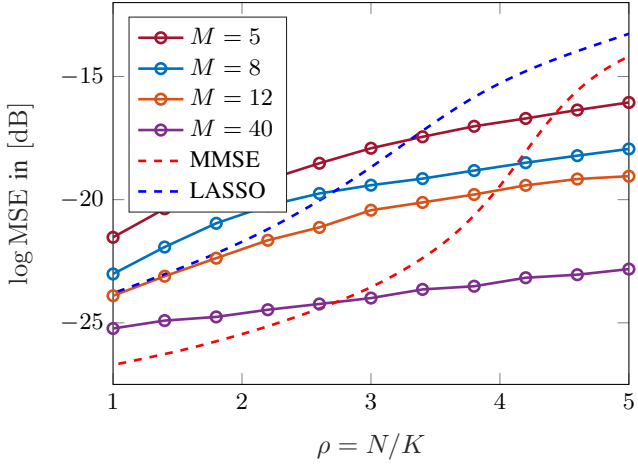}
\caption{\ac{mse} vs inverse load for Algorithm~\ref{alg2}. Here, $N=200$ sparse Gaussian samples with sparsity $\delta=0.1$ are sensed by \ac{iid} random matrices. The \ac{oas} is adapted via hard thresholding~with $\log D_{\rm Th} =-26.5$ $\rm dB$, and $\sigma^2 = 0.01$. The dashed lines denote asymptotic performance for non-adaptive LASSO and conditional mean estimation.\vspace*{-3mm}}
\label{Fig:3}
\end{figure}
\begin{algorithm}[t]
\caption{\ac{oas} via Matched Filtering}
\label{alg2}
\begin{algorithmic}[0]
\Initiate 
$d_n \brc{\mY_0} = +\infty$,
$\bar{\yy}_n [0] =0$,
$\sigma_n^2 \brc{0} =0$,
$c_n [0] =0$.\vspace*{1mm}
\For{$m=1:M$}\\
Generate $\mA_m\in\setR^{K\times N}$ \ac{iid} with $\man\brc{0,1/K}$.\\
Set $\mW_m=\set{0}^{N\times N}$ and $\setM_m\subseteq [N]$, such that
\begin{align}
d_{n} \brc{\mY_{m-1}} \geq D_{\Thr}.
\end{align}
for all $n\in\setM_m$. Set $\rho_m = \abs{\setM_m}/K$.\vspace*{1mm}
\For{$n\in\setM_m$}\\
$\qquad$Set
$\mW_m\brc{n,n} = 1$ and $c_{n}[m] = c_{n}[m-1]+1$.
\EndFor\vspace*{1mm}\\
Sense $\bx$ via $\mA_m\mW_m$, i.e. $\by_m = \mA_m\mW_m \bx + \bz_m$.\vspace*{.5mm}
\For{$n\in \setM_m$} update 
\begin{subequations}
\begin{align}
\bar{\yy}_{n} [m] &= \bar{\yy}_{n} [m-1] + \frac{\ba_n^\trp\brc{m} \by_m}{\norm{\ba_n \brc{m}}^2},\\
\sigma_n^2 \brc{m} &= \sigma_n^2 \brc{m-1} + \frac{\rho_m \sigma_t^2 + M \sigma^2}{\norm{\ba_n \brc{m}}^2}\\
\ee_{n}\brc{x \vert \bar{\yy}_{n} [m] } &= \exp\set{-\dfrac{\brc{ \bar{\yy}_{n} [m] - c_{n}[m] x }^2 }{2\sigma_n^2 \brc{m}}}.
\end{align}
\end{subequations}
$\qquad$Update recoveries and conditional distortions as
\begin{subequations}
\begin{align}
r_{n} \brc{\mY_m} = \dfrac{\int x \  \ee_{n}\brc{x \vert \bar{\yy}_{n} [m] } q\brc{x} \dif x }{\int \ee_{n}\brc{x \vert \bar{\yy}_{n} [m] } q\brc{x} \dif x}.\\
d_{n} \brc{\mY_m} = \dfrac{ \partial}{\partial \bar{\yy}_{n} [m]} r_{n} \brc{\mY_m}.
\end{align}
\end{subequations}
\EndFor
\EndFor
\end{algorithmic}
\end{algorithm}

\section{Conclusion}
\label{conclusion}
Asymptotics of Bayesian \ac{oas} with generic random~projections was left unaddressed, due to the computational~intractability. This study has characterized the performance of \ac{oas} in the large-system limit by means of the decoupling property of Bayesian estimators. The results have depicted significant enhancement achieved by \ac{oas} with hard adaptation compared to the non-adaptive \ac{mmse} bound. For sake of implementation, two computationally tractable algorithms based on orthogonal sensing and matched filtering have been proposed. These algorithms outperform non-adaptive LASSO and the \ac{mmse} bound even for small oversampling factors.

Analytic derivations of this study can be further employed to investigate impacts of different adaptation strategies on the performance of \ac{oas}. Design of low-complexity \ac{oas} algorithms based on $\ell_1$-norm minimization is another direction for  future work.

\bibliography{ref}
\bibliographystyle{IEEEtran}
\end{document}